\RequirePackage[OT1]{fontenc}
\documentclass[letterpaper, 10 pt, conference]{ieeeconf} 

\IEEEoverridecommandlockouts  
\usepackage{amsmath}
\usepackage{amssymb}
\usepackage{amsthm}
\usepackage{mathtools}
\usepackage{xcolor}

\usepackage[noadjust]{cite}

\newtheorem{theorem}{Theorem}
\newtheorem{lemma}{Lemma}

\theoremstyle{definition}
\newtheorem{definition}{Definition}
\newtheorem{assumption}{Assumption}
\newtheorem{remark}{Remark}
\newtheorem{problem}{Problem}

\newcommand{\subjectto}{\operatorname{subject~to}}
\newcommand{\R}{\mathbb{R}}

\newcommand{\T}{^\top}
\newcommand{\KL}{\mathcal{KL}}

\newcommand{\Kinf}{\mathcal{K}_{\infty}}

\title{\bf{Modular Adaptive Safety-Critical Control}}

\author{Max H. Cohen and Calin Belta %
\thanks{The authors are with the Department of Mechanical Engineering, Boston University, Boston, MA, USA \{\texttt{maxcohen}, \texttt{cbelta}\}@\texttt{bu.edu}. This work is supported by the NSF under grants DGE-1840990
and IIS-2024606. Any opinions, findings, conclusions or recommendations expressed in this material are those of the author(s) and do not necessarily reflect the views of the NSF.}
}

\begin{document}
    \maketitle

    \begin{abstract}
        This paper presents an adaptive control approach for uncertain nonlinear systems subject to safety constraints that allows for modularity in the selection of the parameter estimation algorithm. Such modularity is achieved by unifying the concepts of input-to-state stability (ISS) and input-to-state safety (ISSf) via control Lyapunov functions (CLFs) and control barrier functions (CBFs), respectively. In particular, we propose a class of exponential ISS-CLFs and ISSf high order CBFs that can be combined with a general class of parameter estimation algorithms akin to those found in the literature on concurrent learning adaptive control. We demonstrate that the unification of ISS and ISSf in an adaptive control setting allows for maintaining a single set of parameter estimates for both the CLF and CBF that can be generated by a class of update laws satisfying a few general properties. The modularity of our approach is demonstrated via numerical examples by comparing performance in terms of stability and safety across different parameter estimation algorithms.
    \end{abstract}

    \section{Introduction}
    Adaptive control theory \cite{Krstic,Slotine} is concerned with simultaneous learning and control of uncertain dynamical systems. In traditional adaptive control, learning often manifests itself as the estimation of uncertain parameters associated with the underlying dynamical system, whereas control is synonymous with stabilization to a set point or tracking of a desired reference trajectory. In the context of adaptive control of nonlinear systems, designs are often classified as either \emph{Lyapunov-based} \cite{KrsticSCL95} or \emph{modular} \cite{KrsticAutomatica96}. Lyapunov-based designs \cite[Ch. 3-4]{Krstic} typically rely on the design of a control Lyapunov function (CLF) \cite{SontagSCL89} for a modified dynamical system, where potentially destabilizing parameter estimation errors are eliminated through the use of a Lyapunov-based parameter update law. Such designs have the benefit of guaranteeing asymptotic stability (or even exponential stability \cite{ChowdharyCDC10,ChowdharyACC11,ChowdharyIJACSP13,KamalapurkarTAC17,DixonIJACSP19}), but restrict the design of the parameter estimation algorithm since it is tightly coupled to the associated CLF. On the other hand, modular designs \cite[Ch. 5-6]{Krstic} decouple the design of controller and update law: a controller and parameter estimator satisfying a few general properties can be combined to enforce weaker forms of stability, such as input-to-state stability (ISS) \cite{SontagTAC89}. Allowing modularity in the estimation algorithm is motivated by the fact that various estimation procedures may provide certain benefits compared to Lyapunov-based update laws. For example, least-squares based estimation algorithms generally exhibit faster convergence than gradient descent based estimation algorithms employed in Lyapunov-based approaches. 
    
    Motivated by the need for certifiably correct behavior of modern autonomous systems, adaptive control techniques have recently been applied to more complicated control problems, such as guaranteeing safety \cite{TaylorACC20,LopezLCSS21,HovakimyanCDC20,DixonACC21,PanagouECC21,AzimiTAC21,CohenACC22}, often formalized using set theoretic notions \cite{Blanchini}, and enforcing more general temporal logic specifications \cite{SadraCDC17,CohenNAHS23,OzayHSCC22}. In \cite{TaylorACC20}, classical Lyapunov-based adaptive control designs are extended to a safety-critical setting, where a control barrier function (CBF) \cite{AmesTAC17} is designed for a modified dynamical system and a CBF-based adaptive update law is leveraged to eliminate parameter estimation errors that could otherwise lead to safety violation. Although \cite{TaylorACC20} provides a foundation for extending traditional nonlinear adaptive control designs to enforcing safety using CBFs, the particular approach taken therein is conservative in the sense that the proposed method restricts the system to the nonnegative superlevel sets of the safe set, rather than only the zero superlevel set as is common when using CBFs \cite{LopezLCSS21}. By leveraging known bounds on the system's uncertain parameters, works such as \cite{LopezLCSS21,HovakimyanCDC20,DixonACC21,PanagouECC21,AzimiTAC21} reduce such conservatism by taking a ``robust adaptive'' approach whereby safety is guaranteed by accounting for the worst-case parameter estimation error, which is reduced online as more data about the system is collected. In \cite{CohenACC22} such robust adaptive approaches are extended to CBFs with high relative degree \cite{SreenathACC16,WeiTAC21-hocbf,DimosTAC22-hocbf,PanagouCDC21} in which the control input may not directly influence the derivative of the CBF candidate. 

    The aforementioned approaches to enforcing safety using adaptive control techniques have demonstrated success on a wide variety of problems; however, they suffer from a combination of the following limitations: 1) they restrict the update laws/estimation algorithms that can be used to guarantee safety \cite{TaylorACC20,LopezLCSS21,HovakimyanCDC20,DixonACC21,PanagouECC21,AzimiTAC21,CohenACC22}; 2) they require precise knowledge of the bounds on the system parameters/uncertainties \cite{LopezLCSS21,HovakimyanCDC20,DixonACC21,PanagouECC21,AzimiTAC21,CohenACC22}; 3) they require redundant parameter estimation in the sense that multiple estimates of the same parameters are needed if the adaptive safety controller is combined with an adaptive stabilizing controller to achieve a performance objective \cite{TaylorACC20,LopezLCSS21,DixonACC21,CohenACC22}. The main objective of this paper is to develop a framework for modular adaptive safety-critical control that addresses the previous limitations by 1) allowing freedom in the selection of the parameter update law/estimation algorithm, 2) not requiring precise knowledge of bounds on the system parameters/uncertainties, and 3) allowing for the use of a single set of parameter estimates that are shared between the safety and performance (stabilizing) controller. We accomplish this objective by unifying the concept of ISS with that of \emph{input-to-state safety} (ISSf) \cite{AmesLCSS19,AmesLCSS22-ISSf,AmesArXiV-ISSf}, extending ideas from traditional modular nonlinear adaptive control \cite{KrsticAutomatica96} to a safety-critical setting. In particular, we show how a general class of parameter estimation algorithms can be combined with a particular class of ISS-CLFs \cite{AmesACC18-ISS-CLF} and a class of ISSf \emph{high order} CBFs to simultaneously guarantee ISS and ISSf of the underlying system in a modular fashion. The drawback of this modularity is that we establish safety using an ISSf framework, which studies the invariance of inflated safe sets parameterized by the magnitude of a disturbance (parameter estimation error) perturbing the nominal system dynamics.

    The contributions of this paper are threefold. First, we present a modular approach to nonlinear adaptive stabilization using a class of exponential ISS-CLFs \cite{AmesACC18-ISS-CLF}. We combine such CLFs with a class of parameter estimators characteristic of those found in \emph{concurrent learning} adaptive control \cite{ChowdharyCDC10,ChowdharyACC11,ChowdharyIJACSP13,KamalapurkarTAC17,DixonIJACSP19} and show that our approach guarantees ISS and, under suitable excitation conditions outlined in \cite{ChowdharyCDC10,ChowdharyACC11,ChowdharyIJACSP13,KamalapurkarTAC17,DixonIJACSP19}, asymptotic stability of the origin. Second, we generalize the class of ISSf-CBFs introduced in \cite{AmesLCSS19,AmesLCSS22-ISSf,AmesArXiV-ISSf} to high relative degree safety constraints \cite{SreenathACC16,WeiTAC21-hocbf,DimosTAC22-hocbf,PanagouCDC21}, which allows for the systematic construction of a candidate safe set from a given safety constraint on the lower order system dynamics. Finally, we demonstrate the versatility of our approach via numerical examples in which the uncertain parameters of a system are learned online using different update laws while guaranteeing ISS and ISSf.

    The remainder of this paper is organized as follows. Sec. \ref{sec:prelim} covers preliminaries on ISS and ISSf. Sec. \ref{sec:stability} presents our approach to modular adaptive stabilization. Sec. \ref{sec:safety} introduces the notion of an ISSf high order CBF. Sec. \ref{sec:sim} presents simulations and Sec. \ref{sec:conclusion} contains concluding remarks.

    \section{Preliminaries and Problem Formulation}\label{sec:prelim}
    \paragraph*{Notation}
    A continuous function $\alpha\,:\,\R_{\geq0}\rightarrow\R_{\geq0}$ is said to be a class $\mathcal{K}_{\infty}$ function, denoted by $\alpha\in\mathcal{K}_{\infty}$, if $\alpha(0)=0$, $\alpha$ is strictly increasing, and $\lim_{r\rightarrow\infty}\alpha(r)=\infty$. A continuous function $\alpha\,:\,\R\rightarrow\R$ is said to be an extended class $\mathcal{K}_{\infty}$ function, denoted by $\alpha\in\mathcal{K}_{\infty}^e$, if $\alpha(0)=0$, $\alpha$ is strictly increasing, $\lim_{r\rightarrow\infty}\alpha(r)=\infty$, and $\lim_{r\rightarrow-\infty}\alpha(r)=-\infty$. A continuous function $\beta\,:\,\R_{\geq0}\times\R_{\geq0}\rightarrow\R_{\geq0}$ is said to be a class $\mathcal{KL}_{\infty}$ function, denoted by $\beta\in \mathcal{KL}_{\infty}$, if $\beta(\cdot,s)\in\mathcal{K}_{\infty}$ for all $s\in\R_{\geq0}$, $\beta(r,\cdot)$ is decreasing in its second argument, and $\lim_{s\rightarrow\infty}\beta(r,s)=0$ for all $r\in\R_{\geq0}$. The Euclidean norm is denoted by $\|\cdot\|$. The spaces of bounded and square integrable functions are denoted by $\mathcal{L}_{\infty}$ and $\mathcal{L}_2$, respectively. Given a bounded and piecewise continuous function $d\,:\,\R_{\geq0}\rightarrow\R^m$, its $\mathcal{L}_{\infty}$ norm is denoted by $\|d\|_{\infty}=\sup_{t\in\R_{\geq0}}\|d(t)\|<\infty$. Given a continuously differentiable scalar function $h\,:\,\R^n\rightarrow\R$ and a vector field $f\,:\,\R^n\rightarrow\R^n$, the Lie derivative of $h$ along $f$ is denoted by $L_fh(x)=\nabla h(x)f(x)$, where $\nabla h\,:\,\R^n\rightarrow\R^{1\times n}$ is the gradient of $h$. The boundary and interior of a closed set $\mathcal{C}$ are denoted by $\partial\mathcal{C}$ and $\text{Int}(\mathcal{C})$.

    \subsection{Input-to-state stability}\label{sec:ISS}
    Consider the uncertain nonlinear control affine system
    \begin{equation}\label{eq:dyn}
        \dot{x} = f(x) + F(x)\theta +  g(x)u,
    \end{equation}
    where $x\in\R^n$ is the system state, $u\in\R^m$ is the control input, and $\theta\in\R^p$ is a vector of uncertain parameters. The vector field $f\,:\,\R^n\rightarrow\R^n$ captures the system drift dynamics, the columns of $g\,:\,\R^n\rightarrow\R^{n\times m}$ represent vector fields describing the control directions, and the matrix-valued function $F\,:\,\R^{n}\rightarrow\R^{n\times p}$ is a known basis of nonlinear features for the uncertain parameters. We assume $f$, $g$, and $F$, are locally Lipschitz continuous, and that $f(0)=0$ and $F(0)=0$ so that the origin is an equilibrium point of \eqref{eq:dyn} with $u=0$. Our main objective is to learn the uncertain parameters in \eqref{eq:dyn} completely online, while guaranteeing stability and safety. To this end, let $\hat{\theta}\in\R^p$ be an estimate of $\theta$ (the true unknown value of the model parameters) and define $\tilde{\theta}\coloneqq\theta - \hat{\theta}$ as the parameter estimation error. Using $\tilde{\theta}$, system \eqref{eq:dyn} is equivalent to
    \begin{equation}\label{eq:dyn-iss}
        \dot{x} = f(x) + F(x)\hat{\theta} + g(x)u + F(x)\tilde{\theta},
    \end{equation}
    and given a controller $k\,:\,\R^n\times\R^p\rightarrow\R^m$, locally Lipschitz in both its arguments, we obtain the closed-loop system by fixing $u=k(x,\hat{\theta})$ as
    \begin{equation}\label{eq:dyn-cl}
        \dot{x} = f(x) + F(x)\hat{\theta} + g(x)k(x,\hat{\theta}) + F(x)\tilde{\theta}.
    \end{equation}
    When $\tilde{\theta}$ is viewed as a disturbance input to the nominal closed-loop dynamics $f(x) + F(x)\hat{\theta} + g(x)k(x,\hat{\theta})$, an elegant tool for studying the stability of \eqref{eq:dyn-cl} is the notion of \emph{input-to-state stability}.
    \begin{definition}[\cite{Krstic}]\label{def:ISS}
        System \eqref{eq:dyn-cl} is said to be \emph{input-to-state stable} (ISS) if for each initial condition $x(0)\in\R^n$ and each $\tilde{\theta}(\cdot)\in\mathcal{L}_{\infty}$ the trajectory of the closed-loop system \eqref{eq:dyn-cl} satisfies
        \begin{equation}
            \|x(t)\| \leq \beta(\|x(0)\|,t) + \iota(\|\tilde{\theta}\|_{\infty}),\quad \forall t\in\R_{\geq0},
        \end{equation}
        for some $\beta\in\KL$ and $\iota\in\Kinf$. If $\beta(r,s)= c r e^{-\lambda s}$ for some positive constants $c,\lambda\in\R_{>0}$, then \eqref{eq:dyn-iss} is said to be \emph{exponentially ISS} (eISS).
    \end{definition}

    \subsection{Input-to-state safety}\label{sec:ISSf}
    In addition to stability, we are interested in studying the safety properties of \eqref{eq:dyn}, which is often associated with the concept of set invariance \cite{Blanchini}. Given a locally Lipschitz feedback controller for \eqref{eq:dyn}, a set
    \begin{equation}\label{eq:C}
        \mathcal{C} = \{x\in\R^n\,|\,h(x)\geq 0\},
    \end{equation}
    where $h\,:\,\R^n\rightarrow\R$ is continuously differentiable, is said to be \emph{forward invariant} if the resulting solution $x\,:\,I(x(0))\rightarrow\R^n$ satisfies $x(t)\in\mathcal{C}$ for all $t\in I(x(0))$, where $I(x(0))\subseteq \R_{\geq 0}$ is the solution's maximal interval of existence from an initial condition of $x(0)\in\R^n$. When studying the disturbed system \eqref{eq:dyn-iss}, the notion of \emph{input-to-state safety} (ISSf), introduced in \cite{AmesLCSS19,AmesLCSS22-ISSf,AmesArXiV-ISSf}, provides a methodology to study the impact of uncertainties on safety. In particular, the ISSf framework is concerned with establishing the forward invariance of an inflated version of \eqref{eq:C} defined as
    \begin{equation}\label{eq:C-ISSf}
        \begin{aligned}
            \mathcal{C}_\delta = &\{x\in\R^n\,|\,h(x) + \gamma(\delta) \geq 0\},
        \end{aligned}
    \end{equation}
   where $\gamma\in\mathcal{K}_{\infty}$.

    \begin{definition}[\cite{AmesArXiV-ISSf}]
        System \eqref{eq:dyn-cl} is said to be \emph{input-to-state safe} (ISSf) on a set $\mathcal{C}$ as in \eqref{eq:C} if there exists a
        $\gamma$ as in \eqref{eq:C-ISSf} such that for all $\delta\in\R_{\geq0}$ and all $\tilde{\theta}(\cdot)\in\mathcal{L}_{\infty}$ satisfying $\|\tilde{\theta}\|_{\infty}\leq \delta$, the set $\mathcal{C}_\delta$ defined as in \eqref{eq:C-ISSf} is forward invariant.
    \end{definition}
    
    \begin{problem}
        Given system \eqref{eq:dyn-iss} and a set $\mathcal{C}\subset\R^n$, construct a control policy $u=k(x,\hat{\theta})$ and a parameter update law $\dot{\hat{\theta}}$ such that the closed-loop system is ISS and ISSf on $\mathcal{C}$.
    \end{problem}

    \section{Modular Adaptive ISS}\label{sec:stability}
    In this section we introduce a modular adaptive control approach to stabilization by exploiting a class of exponential ISS control Lyapunov functions (eISS-CLF) \cite{AmesACC18-ISS-CLF}. Although our development leverages ISS-CLFs, we illustrate that, under suitable assumptions, the controller derived from this ISS-CLF guarantees asymptotic stability of the closed-loop system, rather than ISS. In principle our approach can be used with any parameter estimation algorithm that guarantees boundedness of the estimation error; however, we specialize our results to a particular class of parameter estimators whose characteristics are outlined in the following lemma.
    \begin{lemma}\label{lemma:estimator}
        Consider a parameter update law $\dot{\hat{\theta}}=\tau(\hat{\theta},t)$, with $\tau$ locally Lipschitz in its first argument and piecewise continuous in its second, and a Lyapunov-like function $V_{\theta}\,:\,\R^p\times \R_{\geq 0}\rightarrow\R_{\geq 0}$, continuously differentiable in both its arguments, satisfying
        \begin{equation}\label{eq:update_law1}
            \eta_1\|\tilde{\theta}\|^2 \leq V_{\theta}(\tilde{\theta},t) \leq \eta_2\|\tilde{\theta}\|^2,\quad \forall (\tilde{\theta},t)\in\R^p\times\R_{\geq0},
        \end{equation}
        for some $\eta_1,\eta_2\in\R_{>0}$. Provided 
        \begin{equation}\label{eq:update_law2}
            \dot{V}_{\theta}(\tilde{\theta},t)\leq 0, \quad \forall (\tilde{\theta},t)\in\R^p\times\R_{\geq0},
        \end{equation}
        then $\tilde{\theta}(\cdot)\in\mathcal{L}_{\infty}$. Furthermore, if there exists a pair $(\eta_3,T)\in\R_{>0}\times\R_{\geq0}$ such that 
        \begin{equation}\label{eq:update_law3}
            \dot{V}_{\theta}(\tilde{\theta},t) \leq -\eta_3\|\tilde{\theta} \|^2, \quad \forall (\tilde{\theta},t)\in\R^p\times\R_{\geq T},
        \end{equation}
        then $\tilde{\theta}(\cdot)\in\mathcal{L}_2\cap\mathcal{L}_{\infty}$ and 
        \begin{equation}\label{eq:theta_bound}
            \|\tilde{\theta}(t)\| \leq \frac{\eta_2}{\eta_1}\|\tilde{\theta}(0)\|e^{-\frac{\eta_3}{2 \eta_2}(t- T)}, \quad \forall t\in\R_{\geq0}.
        \end{equation}
    \end{lemma}
    \begin{proof}
        Since $\dot{V}_{\theta}(\tilde{\theta}(t),t)\leq0$ for all $t\in\R_{\geq0}$, $V_{\theta}(\tilde{\theta}(t),t)$ is nonincreasing and $V_{\theta}(\tilde{\theta}(t),t) \leq V_{\theta}(\tilde{\theta}(0),0)$ for all $t\in\R_{\geq 0}$. Using \eqref{eq:update_law1} this implies that for all $t\in\R_{\geq0}$
        \begin{equation}\label{eq:theta_bound1}
            \|\tilde{\theta}(t)\|\leq \sqrt{\tfrac{\eta_2}{\eta_1}}\|\tilde{\theta}(0)\|,
        \end{equation}
        and thus $\tilde{\theta}(\cdot)\in\mathcal{L}_{\infty}$. For $t\in\R_{\geq T}$, $\dot{V}_{\theta}(\tilde{\theta}(t),t)\leq-\eta_3\|\tilde{\theta}(t)\|^2$, which, after using the comparison lemma \cite[Lemma 3.4]{Khalil} and \eqref{eq:update_law1}, implies
        \begin{equation}
            \begin{aligned}
                \|\tilde{\theta}(t)\| \leq & \sqrt{\tfrac{\eta_2}{\eta_1}}\|\tilde{\theta}(T)\| e^{-\frac{\eta_3}{2\eta_2}(t- T)} \\ 
                \leq & \tfrac{\eta_2}{\eta_1}\|\tilde{\theta}(0)\|e^{\frac{\eta_3}{2\eta_2}T}e^{-\frac{\eta_3}{2\eta_2}t},
            \end{aligned}
        \end{equation}
        for all $t\in \R_{\geq T}$. This bound is also valid for all $t\in\R_{\geq0}$ as stated in \eqref{eq:theta_bound} since $1 \leq \exp(-\tfrac{\eta_3}{2\eta_2}(t-T))$ for all $t\in[0,T]$. The bound in \eqref{eq:theta_bound} also implies $\tilde{\theta}(\cdot)\in\mathcal{L}_2$ since
        \begin{equation*}
            \begin{aligned}
                \int_{0}^t\|\tilde{\theta}(s)\|^2ds\leq & (\tfrac{\eta_2}{\eta_1}\|\tilde{\theta}(0)\|e^{\frac{\eta_3}{2\eta_2}T})^2\int_{0}^{t}e^{-\frac{\eta_3}{\eta_2}s}ds \\
                = & (\tfrac{\eta_2}{\eta_1}\|\tilde{\theta}(0)\|e^{\frac{\eta_3}{2\eta_2}T})^2(-\tfrac{\eta_2}{\eta_3}(e^{-\frac{\eta_3}{\eta_2}t} - 1)),
            \end{aligned}
        \end{equation*}
        which, after taking limits as $t\rightarrow\infty$, implies that
        \begin{equation}\label{eq:theta-L2}
            \int_{0}^{\infty}\|\tilde{\theta}(s)\|^2ds\leq \tfrac{\eta_2^3}{\eta_3\eta_1^2}\|\tilde{\theta}(0)\|^2e^{\frac{\eta_3}{\eta_2}T} <\infty.
        \end{equation}
        Combining \eqref{eq:theta-L2} and \eqref{eq:theta_bound1} yields $\tilde{\theta}(\cdot)\in\mathcal{L}_2\cap\mathcal{L}_{\infty}$.
    \end{proof}

    The condition in \eqref{eq:update_law2} requires that the origin of the parameter estimation error dynamics is stable in the sense of Lyapunov -- a property satisfied by standard estimation algorithms (cf. \cite{KrsticAutomatica96}). The condition in \eqref{eq:update_law3} is more restrictive. It asks, after a certain time period, for the parameter estimates to exponentially converge to their true values. Traditionally, this is only guaranteed under prohibitive persistence of excitation (PE) conditions that generally cannot be verified in practice \cite{ChowdharyCDC10}. Over the past decade, however, a suite of tools termed \emph{concurrent learning adaptive control} \cite{ChowdharyCDC10} have emerged that relax the PE condition by maintaining a sufficiently rich ``history stack" of input-output data \cite{ChowdharyACC11}. In this regard, the interval $[0,T)$ in Lemma \ref{lemma:estimator} corresponds to an initial input-output data collection phase; once a sufficiently rich dataset has been collected such data can be exploited to ensure exponential convergence of the parameter estimates (see \cite{ChowdharyACC11} for data collection strategies). We refer the interested reader to \cite{ChowdharyCDC10,ChowdharyACC11,ChowdharyIJACSP13} for a more thorough introduction and to \cite{DixonIJACSP19,KamalapurkarTAC17,KamalapurkarAutomatica16-mbrl,DeptulaAutomatica21} for specific instances of update laws satisfying the conditions of Lemma \ref{lemma:estimator}. The class of parameter estimators outlined in the preceding lemma will be combined with the notion of an eISS-CLF to establish stability of \eqref{eq:dyn-iss} in the presence of uncertain parameters.

    \begin{definition}\label{def:eISS-CLF}
        A continuously differentiable positive definite function $V\,:\,\R^n\rightarrow\R_{\geq 0}$ is said to be an \emph{exponential input-to-state stable control Lyapunov function} (eISS-CLF) for \eqref{eq:dyn-iss} if there exist positive constants $c_1,c_2,c_3,\varepsilon\in\R_{> 0}$ such that
        \begin{subequations}\label{eq:eISS-CLF}
            \begin{equation}\label{eq:eISS-CLF1}
                c_1\|x\|^2 \leq V(x) \leq c_2\|x\|^2,\quad \forall x\in\R^n,
            \end{equation}
            \begin{equation}\label{eq:eISS-CLF2}
                \begin{aligned}
                    \inf_{u\in\R^m}\dot{\hat{V}}(x,\hat{\theta},u)  < -c_3V(x) - \tfrac{1}{\varepsilon}\|L_FV(x)\|^2,
                \end{aligned}
            \end{equation}
        \end{subequations}
        for all $(x,\hat{\theta})\in(\R^n\setminus\{0\})\times\R^p$ where 
        \begin{equation*}
            \dot{\hat{V}}(x,\hat{\theta},u)\coloneqq L_fV(x) + L_FV(x)\hat{\theta} + L_gV(x)u.
        \end{equation*}
    \end{definition}
    When the uncertain parameters in \eqref{eq:dyn-iss} are \emph{matched}\footnote{The parameters \eqref{eq:dyn-iss} are matched when there exists a locally Lipschitz mapping $\varphi\,:\,\R^n\rightarrow\R^{m\times p}$ such that $F(x) = g(x)\varphi(x)$ for all $x\in\R^n$.\label{footnote:matched}} the construction of an eISS-CLF can be performed by constructing a CLF for the nominal dynamics $\dot{x}=f(x) + g(x)u$.
    We show in the following theorem that combining a parameter estimator satisfying the conditions of Lemma \ref{lemma:estimator} with a controller satisfying the conditions in \eqref{eq:eISS-CLF} renders the closed-loop system ISS and, under additional conditions, renders the origin asymptotically stable. 

    \begin{theorem}\label{theorem:iss}
        If $V$ is an eISS-CLF for \eqref{eq:dyn-iss} and the conditions of Lemma \ref{lemma:estimator} hold, then any controller $u=k(x,\hat{\theta})$ locally Lipschitz on $(\R^n\setminus\{0\})\times \R^p$ satisfying
        \begin{equation*}
            \dot{\hat{V}}(x,\hat{\theta},k(x,\hat{\theta})) \leq -c_3V(x) - \tfrac{1}{\varepsilon}\|L_FV(x)\|^2,
        \end{equation*}
        for all $(x,\hat{\theta})\in\R^{n}\times\R^p$ renders \eqref{eq:dyn-iss} eISS.
        Furthermore, $\lim_{t\rightarrow\infty}x(t)=0$.
    \end{theorem}

    \begin{proof}
        The time derivative of $V$ can be bounded as
        \begin{equation}\label{eq:Vdot}
            \begin{aligned}
                \dot{V} = & L_fV(x) + L_FV(x)\hat{\theta} + L_gV(x)k(x,\hat{\theta}) + L_FV(x)\tilde{\theta} \\ 
                \leq & -c_3V(x) - \tfrac{1}{\varepsilon}\|L_FV(x)\|^2 + L_FV(x)\tilde{\theta} \\ 
                \leq & -c_3V(x) + \tfrac{\varepsilon}{4}\|\tilde{\theta}\|^2 \\ 
                \leq & -c_3V(x) + \tfrac{\varepsilon}{4}\|\tilde{\theta}\|_{\infty}^2, \\ 
            \end{aligned}
        \end{equation}
        where the first inequality follows from $k$, the second from completing squares, and the third from $\tilde{\theta}(\cdot)\in\mathcal{L}_{\infty}$. Invoking the comparison lemma \cite[Lem. 3.4]{Khalil} and using \eqref{eq:eISS-CLF1} yields
        \begin{equation}
            \begin{aligned}
                \|x(t)\|\leq & \sqrt{\frac{c_2}{c_1}\|x(0)\|^2e^{-c_3t} + \frac{\varepsilon}{4 c_1 c_3}\|\tilde{\theta}\|^2_{\infty}} \\
                \leq & \sqrt{\frac{c_2}{c_1}}\|x(0)\|e^{-\frac{c_3}{2}t} + \frac{1}{2}\sqrt{\frac{\varepsilon}{c_1 c_3}}\|\tilde{\theta}\|_{\infty},
            \end{aligned}
        \end{equation}
        where the second inequality follows from the fact that the square root is a subadditive function, implying the closed-loop system is eISS and thus $x(\cdot)\in\mathcal{L}_{\infty}$. To show that $\lim_{t\rightarrow \infty} x(t)=0$, we rearrange the third line of \eqref{eq:Vdot} and use \eqref{eq:eISS-CLF1} to obtain
        \begin{equation}
            \begin{aligned}
                c_1c_3\|x\|^2\leq c_3V(x) \leq \tfrac{\varepsilon}{4}\|\tilde{\theta}\|^2 - \dot{V}.
            \end{aligned}
        \end{equation}
        Integrating the above over a finite time interval $[0,t]$ yields
        \begin{equation*}\label{eq:x_int1}
            \begin{aligned}
                c_1c_3\int_{0}^t\|x(s)\|^2ds \leq & \frac{\varepsilon}{4}\int_{0}^{t}\|\tilde{\theta}(s)\|^2ds - V(x(t)) + V(x(0)) \\ 
                \leq & \frac{\varepsilon}{4}\int_{0}^{t}\|\tilde{\theta}(s)\|^2ds + V(x(0)).
            \end{aligned}
        \end{equation*}
        Taking limits as $t\rightarrow\infty$ and noting that $\tilde{\theta}(\cdot)\in\mathcal{L}_2$ yields
        \begin{equation*}
            \begin{aligned}
                \int_{0}^{\infty}\|x(s)\|^2ds \leq & \frac{\varepsilon}{4c_1c_3}\int_{0}^{\infty}\|\tilde{\theta}(s)\|^2ds + \tfrac{V(x(0))}{c_1c_3} < \infty, \\ 
            \end{aligned}
        \end{equation*}
       implying $x(\cdot)\in\mathcal{L}_2$. It follows from Lemma \ref{lemma:estimator} that $\tilde{\theta}(\cdot)\in\mathcal{L}_{\infty}$ and thus $\hat{\theta}(\cdot)\in\mathcal{L}_{\infty}$. Combining this with the assumption that $u=k(x,\hat{\theta})$ is locally Lipschitz in both its arguments implies that $u(\cdot)\in\mathcal{L}_{\infty}$ and thus $\dot{x}(\cdot)\in\mathcal{L}_{\infty}$. Since $x(\cdot),\dot{x}(\cdot)\in\mathcal{L}_{\infty}$ and $x(\cdot)\in\mathcal{L}_2$, Barbalat's Lemma \cite[Corollary A.7]{Krstic} implies $\lim_{t\rightarrow\infty}x(t)=0$.
    \end{proof} 

    Given an eISS-CLF as in Def. \ref{def:eISS-CLF}, inputs satisfying the conditions of Theorem \ref{theorem:iss} can be computed for any $(x,\hat{\theta})$ by solving the quadratic program (QP)
    \begin{equation}\label{eq:eISS-CLF-QP}
        \begin{aligned}
            \min_{u\in\R^m}\quad & \tfrac{1}{2}u\T u \\
            \subjectto \quad & \dot{\hat{V}}(x,\hat{\theta},u) \leq -c_3 V(x) - \tfrac{1}{\varepsilon}\|L_FV(x)\|^2.
        \end{aligned}
    \end{equation}
    As noted in \cite{JankovicAutomatica18,AmesArXiV-ISSf}, the strict inequality in \eqref{eq:eISS-CLF} helps to establish Lipschitz continuity of the QP-based controller away from the origin and is independent of the non-strict inequality that a particular controller must satisfy to guarantee stability. This observation also applies to the CBFs outlined in the next section. Continuity at the origin can be ensured provided the eISS-CLF $V$ satisfies the \emph{small control property} (see \cite{JankovicAutomatica18} for further details). 

    \section{Modular Adaptive ISSf}\label{sec:safety}
    In this section we shift our attention to the problem of establishing safety of \eqref{eq:dyn-iss} in the presence of uncertain parameters. Importantly, we aim to establish such safety guarantees without having to commit to any particular parameter update law as in \cite{TaylorACC20,LopezLCSS21,HovakimyanCDC20,DixonACC21,PanagouECC21,AzimiTAC21,CohenACC22}. We accomplish this by unifying the ISSf framework \cite{AmesLCSS19,AmesLCSS22-ISSf,AmesArXiV-ISSf} with the high order CBF (HOCBF) framework \cite{SreenathACC16,WeiTAC21-hocbf,DimosTAC22-hocbf,PanagouCDC21}, which allows one to recursively compute a candidate safe set from a user-defined high relative degree safety constraint. We note that efforts towards this unification have been explored by \cite{XuArXiV22} in the context of safety \emph{verification} of interconnected systems. Here, we present a formulation better suited for control synthesis in the context of uncertain systems.
    \begin{definition}[\cite{PanagouCDC21}]\label{def:relative-degree}
        A function $h\,:\,\R^n\rightarrow\R$ is said to have relative degree $r\in\mathbb{N}$ for \eqref{eq:dyn-iss} with respect to $u$ on a domain $\mathcal{D}\subseteq\R^n$ if 
        \begin{enumerate}
            \item $h$ is $r$-times differentiable;
            \item for all $x\in\R^n$ and all $i\in\{0,1,\dots,r-2\}$, we have $L_gL_f^ih(x)=0$;
            \item $L_gL_f^{r-1}h(x)\neq 0$ for all $x\in\mathcal{D}$.
        \end{enumerate}
    \end{definition}
    The relative degree of a function $h$ for \eqref{eq:dyn-iss} with respect to $\tilde{\theta}$ is defined similarly by replacing $g$ with $F$ in Def. \ref{def:relative-degree}. Now consider a function $h\,:\,\R^n\rightarrow\R$ of relative degree $r\in\mathbb{N}$ for \eqref{eq:dyn-iss} with respect to $u$ and define
    \begin{equation}\label{eq:C1}
        \mathcal{S}\coloneqq \{x\in\R^n\,|\,h(x)\geq 0\},
    \end{equation}
    as the \emph{constraint set} that we desire to render invariant.
    \begin{assumption}[\cite{CohenACC22}]\label{assuption:relative-degree}
        The constraint function $h$ in \eqref{eq:C1} has relative degree $r\in\mathbb{N}$ for \eqref{eq:dyn-iss} with respect to both $u$ and $\tilde{\theta}$. 
    \end{assumption}
    The implication of Assumption \ref{assuption:relative-degree} is that \emph{both} the control input and uncertain parameters only appear in the $r$-th total derivative of $h$, which will facilitate the development of affine constraints on the control input that are sufficient to establish ISSf, and is not restrictive provided the uncertain parameters are matched (cf. Footnote \ref{footnote:matched}). Given a constraint function $h$ of relative degree $r\in\mathbb{N}$ as in \eqref{eq:C1} and a collection of sufficiently smooth $\alpha_i\in\mathcal{K}_{\infty}^e$, $i\in\{1,\dots,r\}$, we define the sequence of functions
    \begin{equation}\label{eq:psi}
        \begin{aligned}
            \psi_0(x)\coloneqq & h(x) \\ 
            \psi_i(x)\coloneqq & \dot{\psi}_{i-1}(x) + \alpha_i(\psi_{i-1}(x)),\;\forall i\in\{1,\dots,r-1\},\\
            \psi_r(x,u) =& \dot{\psi}_{r-1}(x,u) + \alpha_r(\psi_{r-1}(x)).
        \end{aligned}
    \end{equation}
    We associate to each $i\in\{1,\dots,r\}$ a set $\mathcal{C}^i\subset\R^n$ defined as the zero-superlevel set of $\psi_{i-1}$ as
    \begin{equation}\label{eq:Ci}
        \mathcal{C}^i\coloneqq\{x\in\R^n\,|\,\psi_{i-1}(x)\geq 0\},
    \end{equation}
    and define a candidate safe set\footnote{Note that $\mathcal{S}=\mathcal{C}^1$ and thus $\mathcal{C}\subset\mathcal{S}$, implying forward invariance of $\mathcal{C}$ is sufficient for satisfaction of the original safety constraint.} as
    \begin{equation}\label{eq:Ccap}
        \mathcal{C}\coloneqq \bigcap_{i=1}^r\mathcal{C}^i.
    \end{equation}
    We now aim to develop a control strategy that renders $\mathcal{C}$ from \eqref{eq:Ccap} ISSf with respect to the parameter estimation error. To this end, we define the sequence of functions
    \begin{equation}\label{eq:rho}
        \rho_{i}(x, \delta) \coloneqq\psi_{i-1}(x) + \gamma_i(\delta)
    \end{equation}
    for all $i\in\{1,\dots,r\}$, where $\psi_{i-1}$ is defined as in \eqref{eq:psi}, and each $\gamma_i\in\Kinf$. Similar to \eqref{eq:Ci}, we associate to each $\rho_i$ a set $\mathcal{C}_\delta^i\subset\R^n$ as
    \begin{equation}\label{eq:Ci_delta}
        \mathcal{C}_\delta^i\coloneqq\{x\in\R^n\,|\,\rho_i(x,\delta)\geq 0\},
    \end{equation}
    and define an inflated version of \eqref{eq:Ccap} as
    \begin{equation}\label{eq:Cd}
        \mathcal{C}_\delta\coloneqq \bigcap_{i=1}^{r}\mathcal{C}_\delta^i,
    \end{equation}
    whose forward invariance will be established using the notion of an \emph{ISSf high order CBF} (ISSf-HOCBF). 

    \begin{definition}\label{def:ISSf-HOCBF}
        A function $h\,:\,\R^n\rightarrow\R$ of relative degree $r$ with respect to $u$ is said to be an \emph{input-to-state safe high order control barrier function} for \eqref{eq:dyn-iss} on $\mathcal{C}_\delta$ defined as in \eqref{eq:Cd} if $\nabla \psi_{i-1}(x)\neq0$ for all $x\in\partial \mathcal{C}^i$ for all $i\in\{1,\dots,r\}$ and there exist sufficiently smooth $\alpha_i\in\mathcal{K}_\infty^e$, $i\in\{1,\dots,r\}$, and a positive constant $\varepsilon\in\R_{>0}$ such that for all $x\in\R^n$ and all $\hat{\theta}\in\R^p$
        \begin{equation}\label{eq:ISSf-HOCBF}
            \begin{aligned}
                \sup_{u\in\R^m}\dot{\hat{\psi}}(x,\hat{\theta},u)> -\alpha_r(\psi_{r-1}(x)) + \frac{\|L_F\psi_{r-1}(x)\|^2}{\varepsilon},
            \end{aligned}
        \end{equation}
        where 
        \begin{equation*}
            \dot{\hat{\psi}}(x,\hat{\theta},u)\coloneqq L_f\psi_{r-1}(x) + L_F\psi_{r-1}(x)\hat{\theta} + L_g\psi_{r-1}(x)u,
        \end{equation*}
        and $\psi_{r-1}$ is defined by $h$ and the choice of $\alpha_i$ from \eqref{eq:psi}.
    \end{definition}
    The requirement that $\nabla \psi_{i-1}(x)\neq 0$ for all $x\in\partial \mathcal{C}^i$ and all $i\in\{1,\dots,r\}$ in the above definition is equivalent to the requirement that $0$ is a regular value of $\psi_{i-1}$ for each $i\in\{1,\dots,r\}$, which is needed for the application of Nagumo's Theorem \cite[Thm. 4.1.28]{AbrahamMarsdenRatiu} to establish forward invariance of the inflated safe set. Given the above definition, we define the pointwise set of all control values satisfying the condition in \eqref{eq:ISSf-HOCBF} as
    \begin{equation*}
        \begin{aligned}
             \underbrace{\bigg\{u\in\R^m\,|\,\dot{\hat{\psi}}(x,\hat{\theta},u) \geq -\alpha_r(\psi_{r-1}(x)) + \tfrac{\|L_F\psi_{r-1}(x)\|^2}{\varepsilon} \bigg\}}_{\eqqcolon K_h(x,\hat{\theta})}
        \end{aligned}
    \end{equation*}
    and show in the following theorem that any locally Lipschitz controller belonging to $K_h$ renders $\mathcal{C}_\delta$ forward invariant.

    \begin{theorem}\label{theorem:issf}
        Let $h$ be an ISSf-HOCBF for \eqref{eq:dyn-iss} on $\mathcal{C}_\delta$ and suppose the conditions of Lemma \ref{lemma:estimator} and Assumption \ref{assuption:relative-degree} hold. Let $k\,:\,\R^n\times\R^p\rightarrow\R^m$ be a locally Lipschitz control policy satisfying $k(x,\hat{\theta})\in K_h(x,\hat{\theta})$ for all $(x,\hat{\theta})\in\R^n\times\R^p$ and suppose $\|\tilde{\theta}\|_{\infty}\leq\delta$. Then, the control policy $u=k(x,\hat{\theta})$ renders $\mathcal{C}_\delta$ forward invariant for the closed-loop system with each $\gamma_i$ defined as
        \begin{equation}\label{eq:gamma}
            \begin{aligned}
                \gamma_r(\delta) \coloneqq & -\alpha_r^{-1}\left(-\tfrac{\varepsilon\delta^2}{4}\right), \\
                \gamma_i(\delta)\coloneqq & -\alpha_{i}^{-1}(-\gamma_{i+1}(\delta)),\;\forall i\in\{1,\dots,r-1\}. \\
            \end{aligned}
        \end{equation}
    \end{theorem}

    \begin{proof}
        Taking the time derivative of $\psi_{r-1}$ yields
        \begin{equation*}
            \begin{aligned}
                \dot{\psi}_{r-1} = & \dot{\hat{\psi}}(x,\hat{\theta},k(x,\hat{\theta})) + L_F\psi_{r-1}(x)\tilde{\theta} \\ 
                \geq & - \alpha_r(\psi_{r-1}(x)) + \tfrac{\|L_F\psi_{r-1}(x)\|^2}{\varepsilon} 
                - \|L_F\psi_{r-1}(x)\| \delta \\
                \geq & - \alpha_r(\psi_{r-1}(x)) - \tfrac{\varepsilon\delta^2}{4}, 
            \end{aligned}
        \end{equation*}
        where the first inequality follows from the definition of $K_h$ and $L_F\psi_{r-1}(x)\tilde{\theta}\geq -\| L_F\psi_{r-1}(x)\|\delta$, and the second from completing squares. Noting that $\dot{\psi}_{r-1}=\dot{\rho}_r$ we also have
        \begin{equation}\label{eq:rho_r_dot}
            \begin{aligned}
                \dot{\rho}_r \geq & - (\alpha_r(\psi_{r-1}(x)) + \tfrac{\varepsilon\delta^2}{4}).
            \end{aligned}
        \end{equation}
        Note that for $x\in\mathcal{C}_\delta$ we have $\rho_r(x,\delta) \geq 0$ and thus $\psi_{r-1}(x)\geq - \gamma_r(\delta) $, which, by the monotonicity of $\alpha_r$ and the definition of $\gamma_r$ from \eqref{eq:gamma}, implies that $\alpha_{r}(\psi_{r-1}(x)) \geq -\varepsilon\delta^2/4$. It then follows from \eqref{eq:rho_r_dot} that $\rho_r$ may decrease for $x\in \mathcal{C}_\delta$, i.e., the system may approach $\partial \mathcal{C}_\delta ^r$.
        Our objective is now to show that $\dot{\rho}_r\geq 0$ for any $x\in \partial \mathcal{C}_\delta^r$, which, along with the assumption that $\nabla\psi_{r-1}(x)\neq0$ for all $x\in\partial \mathcal{C}^r$, is a sufficient condition to establish the forward invariance of $\mathcal{C}_\delta^r$ using Nagumo's Theorem \cite[Thm. 4.1.28]{AbrahamMarsdenRatiu}. According to \eqref{eq:rho_r_dot}, for such a condition to hold it is sufficient to show that
        \begin{equation}\label{eq:gammar_cond}
            \alpha_r(\psi_{r-1}(x)) + {\varepsilon\delta^2}/{4} =0,
        \end{equation}
        for all $x\in\partial \mathcal{C}_\delta^r$. Note that for $x\in\partial \mathcal{C}_\delta ^r$ we have $\rho_r(x,\delta)=0$ and hence $\psi_{r-1}(x)=-\gamma_r(\delta)$, which, after using the definition of $\gamma_r$ from \eqref{eq:gamma}, implies \eqref{eq:gammar_cond} holds
        for all $x\in\partial \mathcal{C}_\delta ^r$. It then follows from Nagumo's Theorem that $x(0)\in\mathcal{C}_{\delta}^r\implies x(t)\in\mathcal{C}_{\delta}^r$ for all $t\in I(x(0))$ and thus $\rho_r(x(t),\delta) \geq 0$ for all $t\in I(x(0))$, which implies that $\psi_{r-1}(x(t))\geq -\gamma_r(\delta)$ for all $t\in I(x(0))$.  We proceed with a similar analysis for the remaining $\rho_i$ terms. Using the definition of $\psi_{r-1}$ from \eqref{eq:psi}, and dropping time-dependence for ease of readability, the preceding argument implies 
        \begin{equation}
            \begin{aligned}
                \dot{\psi}_{r-2} = & -\alpha_{r-1}(\psi_{r-2}(x)) + \psi_{r-1}(x) \\
                \geq & -\alpha_{r-1}(\psi_{r-2}(x)) - \gamma_r(\delta).
            \end{aligned}
        \end{equation}
        Noting that $\dot{\rho}_{r-1}=\dot{\psi}_{r-2}$ we have
        \begin{equation}\label{eq:rho_r1_dot}
            \dot{\rho}_{r-1} \geq -(\alpha_{r-1}(\psi_{r-2}(x)) + \gamma_r(\delta)).
        \end{equation}
        Now for $x\in\mathcal{C}_\delta$ we have $\rho_{r-1}(x,\delta)\geq 0$, which, by the same reasoning as in the case of $\rho_r$, implies that $\alpha_{r-1}(\psi_{r-2}(x))\geq -\alpha_{r-1}(-\gamma_{r-1}(\delta))=-\gamma_{r}(\delta)$, and thus $\rho_{r-1}$ may decrease for $x\in \mathcal{C}_\delta$ by \eqref{eq:rho_r1_dot}. However, for $x\in \partial\mathcal{C}_\delta^{r-1}$ we have $\rho_{r-1}(x,\delta)=\psi_{r-2}(x) + \gamma_{r-1}(\delta)=0$ and it follows from the definition of $\gamma_{r-1}$ \eqref{eq:gamma} that $\dot{\rho}_{r-1}\geq 0$ for any $x\in \partial\mathcal{C}_{\delta}^{r-1}$. It then again follows from Nagumo's Theorem that $x(0)\in\mathcal{C}_\delta^{r-1}\implies x(t)\in\mathcal{C}_{\delta}^{r-1}$ for all $x(t)\in I(x(0))$. Following the proof of \cite[Thm. 3]{WeiTAC21-hocbf}, one can take analogous steps to those outlined above for the remaining $\psi_i$ terms to show that $x(0)\in\mathcal{C}_\delta\implies \rho_i(x(t),\delta)\geq 0$ for all $t\in I(x(0))$ and $i\in\{1,\dots,r\}$, implying the forward invariance of $\mathcal{C}_\delta$.
    \end{proof}

    Given an eISS-CLF and ISSf-HOCBF one can compute inputs satisfying the conditions of Theorem \ref{theorem:issf} and (relaxed) conditions of Theorem \ref{theorem:iss} by solving the QP
    \begin{equation}\label{eq:ISSf-ISS-QP}
        \begin{aligned}
            \min_{u\in\R^m}\quad & \tfrac{1}{2}\|u - k(x,\hat{\theta})\|^2 \\
            \text{s.t.} \quad & \dot{\hat{\psi}}(x,\hat{\theta},u) \geq -\alpha_r(\psi_{r-1}(x)) + \tfrac{\|L_F\psi_{r-1}(x)\|^2}{\varepsilon},
        \end{aligned}
    \end{equation}
    where $k(x,\hat{\theta})$ satisfies the conditions of Theorem \ref{theorem:iss}. Rather than combining the CLF and CBF in a single QP as in \cite{AmesTAC17,JankovicAutomatica18}, here we filter the solution of \eqref{eq:eISS-CLF-QP} through \eqref{eq:ISSf-ISS-QP}. This obviates the need to select an appropriate penalty on the relaxation of the CLF constraint, which can lead to controllers with large Lipschitz constants if chosen improperly \cite{JankovicAutomatica18}.

    \begin{remark}
        In contrast to existing approaches \cite{LopezLCSS21,HovakimyanCDC20,DixonACC21,PanagouECC21,CohenACC22}, implementation of \eqref{eq:ISSf-ISS-QP} does not explicitly require knowledge of bounds on the estimation error, which may be unavailable or overly conservative. Furthermore, a single set of parameter estimates are shared between the safety and stability constraint, whereas the methods developed in \cite{TaylorACC20,LopezLCSS21,CohenACC22} require separate estimates of the same parameters. As noted earlier, the drawback of our modular approach is that \eqref{eq:ISSf-ISS-QP} only guarantees forward invariance of an inflated safe set rather than the original constraint set. If a bound on the estimation error is known (e.g., by using the projection operator \cite[App. E]{Krstic} to bound the parameter estimates), then knowledge of the bound $\|\tilde{\theta}\|_{\infty}\leq\delta$ can be incorporated into the ISSf-HOCBF by adding a safety margin to $h$.
    \end{remark}

    \section{Numerical Examples}\label{sec:sim}
    We consider a simple obstacle avoidance scenario for a planar mobile robot modeled as a double integrator with nonlinear drag effects of the form \cite{LopezACC19}
    \begin{equation}\label{eq:dbl_int}
        \ddot{q} = -D\dot{q}\|\dot{q}\| + u,
    \end{equation}
    where $q=[q_1\;q_2]\T\in\R^2$ denotes the robot's position, $u\in\R^2$ its commanded acceleration, and $D\in\R^{2\times 2}$ a diagonal matrix of damping coefficients. Defining $x\coloneqq [q\T\;\dot{q}\T]\T\in\R^4$ allows \eqref{eq:dbl_int} to be expressed as in \eqref{eq:dyn} with
    \begin{equation}
        \dot{x} = 
        \underbrace{
            \begin{bmatrix}
                \dot{q} \\ 0
            \end{bmatrix}
        }_{f(x)}
        +
        \underbrace{
            \begin{bmatrix}
                0_{2\times 2} \\ \text{diag}(\dot{q}\|\dot{q}\|)
            \end{bmatrix}
        }_{F(x)}
        \underbrace{
            \begin{bmatrix}
                D_1 \\ D_2
            \end{bmatrix}
        }_{\theta}
        +
        \underbrace{
            \begin{bmatrix}
                0_{2\times 2} \\ I_{2\times 2}
            \end{bmatrix}
        }_{g(x)}
        u,
    \end{equation}
    where $0_{2\times 2}\in\R^{2\times 2}$ is a $2\times 2$ matrix of zeros, $I_{2\times 2}$ is a $2\times 2$ identity matrix, $\text{diag}(\cdot)$ constructs a diagonal matrix from a vector, and $D_1,D_2\in\R_{>0}$ are the unknown drag coefficients. Our control objective is to drive \eqref{eq:dbl_int} to the origin while avoiding an obstacle in the workspace and learning the uncertain parameters online. To estimate the uncertain parameters, we leverage a general class of concurrent learning parameter estimation algorithms  \cite{ChowdharyCDC10,ChowdharyIJACSP13} based on the method developed in \cite{DixonIJACSP19}. This method works based on the observation that, along state-control trajectory $(x(\cdot),u(\cdot))$, system \eqref{eq:dyn} can be expressed as
    \begin{equation*}
        \begin{aligned}
            \int_{t - \Delta t}^{t}\dot{x}(s)ds = & \int_{t - \Delta t}^{t}f(x(s))ds + \int_{t - \Delta t}^{t}F(x(s))ds\theta \\ 
            & +  \int_{t - \Delta t}^{t}g(x(s))u(s)ds,
        \end{aligned}
    \end{equation*}
    for all $t\geq \Delta t$, where $\Delta t\in\mathbb{R}_{>0}$ is the length of an integration window. Defining 
    \begin{equation*}
        \begin{aligned}
            \mathcal{Y}(t)\coloneqq & \int_{t - \Delta t}^{t}(\dot{x}(s) - f(x(s)) - g(x(s))u(s) )ds, \\
            \mathcal{F}(t) \coloneqq & \int_{t - \Delta t}^{t}F(x(s))ds
        \end{aligned}
    \end{equation*}
    yields the linear relationship for the uncertain parameters
    \begin{equation}
        \mathcal{Y}(t) = \mathcal{F}(t)\theta.
    \end{equation}
    Despite the appearance of $\dot{x}$ in $\mathcal{Y}$, computing $\mathcal{Y}$ only requires state measurements since $\int_{t - \Delta t}^{t}\dot{x}(s)ds=x(t) - x(t-\Delta t)$. The parameters can then be recursively estimated online by storing values of $\mathcal{Y}$ and $\mathcal{F}$ at run-time in a history stack\footnote{This data is stored in the history stack using the singular value maximizing algorithm from \cite{ChowdharyACC11}, which records data so that the convergence rate of the parameter estimates is always non-decreasing.} $\mathcal{H}=\{(\mathcal{Y}_j,\mathcal{F}_j)\}_{j=1}^N$, which is used to update the parameter estimates to minimize the squared prediction error 
    \begin{equation*}
        E(\hat{\theta}) = \sum_{j=1}^{N}\|\mathcal{Y}_j - \mathcal{F}_j\hat{\theta}\|^2.
    \end{equation*}
    To this end, we consider the following class of update laws
    \begin{equation}\label{eq:theta-hat-dot}
        \dot{\hat{\theta}}=-\Gamma(t)\nabla E(\hat{\theta})\T = \Gamma(t)\sum_{j=1}^N \mathcal{F}_j\T(\mathcal{Y}_j - \mathcal{F}_j\hat{\theta}),
    \end{equation}
    which serves as a general template for particular update laws based on the properties of $\Gamma(\cdot)$ as follows:
    \begin{subequations}\label{eq:Gamma-dot}
        \begin{equation}\label{eq:GD}
            \dot{\Gamma} = 0,
        \end{equation}
        \begin{equation}\label{eq:RLS}
            \dot{\Gamma} = -\Gamma\Bigg[\sum_{j=1}^N \mathcal{F}_j\T\mathcal{F}_j\Bigg]\Gamma,
        \end{equation}
        \begin{equation}\label{eq:RLS-beta}
            \dot{\Gamma} = \beta\Gamma -\Gamma\Bigg[\sum_{j=1}^N \mathcal{F}_j\T\mathcal{F}_j\Bigg]\Gamma,
        \end{equation}
        \begin{equation}\label{eq:RLS-beta-t}
            \dot{\Gamma} = \beta\left[1 - \frac{\|\Gamma\|}{\bar{\Gamma}}\right]\Gamma -\Gamma\Bigg[\sum_{j=1}^N \mathcal{F}_j\T\mathcal{F}_j\Bigg]\Gamma,
    \end{equation}
    \end{subequations}
    where $\beta\in\R_{>0}$ is a forgetting/discount factor and $\bar{\Gamma}\in\R_{>0}$ is a user-defined constant that bounds $\|\Gamma(t)\|$. With $\dot{\Gamma}$ as in \eqref{eq:Gamma-dot}, the update law in \eqref{eq:theta-hat-dot} corresponds to: \eqref{eq:GD} gradient descent; \eqref{eq:RLS} recursive least squares (RLS); \eqref{eq:RLS-beta} RLS with a forgetting/discount factor; \eqref{eq:RLS-beta-t} RLS with a variable forgetting factor. An overview of these online parameter estimation algorithms, including their benefits and drawbacks, can be found in \cite[Ch. 8.7]{Slotine}. We emphasize that the purpose of our numerical example is not necessarily to establish superiority of one algorithm over the others; rather, our goal is to demonstrate that, under the assumptions posed in Lemma \ref{lemma:estimator}, the stability/safety guarantees of the controller can be decoupled from the design of the parameter estimator, which allows considerable freedom in selecting an estimation algorithm best suited for the problem at hand. 
    
    We demonstrate the modularity of our approach (i.e., the ability to decouple the design of the estimator from the controller) by running a set of simulations with randomly sampled initial conditions for the system state and estimated parameters under each algorithm, and show that, for a given level of uncertainty, the ISSf guarantees are invariant to the particular choice of parameter estimator. For each estimation algorithm we produce 25 different trajectories by uniformly sampling the initial state from $[-1.8,-2.2]\times[1.8,2.2]\times\{0\}\times\{0\}\subset\R^4$ and the initial parameter estimates from $[0,3]^2\subset\R^2$; the true parameters are set to $\theta=[0.8\;1.4]\T$. The hyperparameters for the estimation algorithms are selected as $N=20$, $\Gamma(0)=100I_{2\times 2}$, $\beta=1$, $\bar{\Gamma}=1000$. The stabilization objective is achieved by considering the eISS-CLF candidate $V(x) = \tfrac{1}{2}\|q\|^2 + \tfrac{1}{2}\|q + \dot{q}\|^2$ with $c_3=1$ and $\varepsilon_V=20$. The safety objective is achieved by considering the constraint function $h(x)=\|q - q_o\|^2 - R_o^2$, where $q_o=[-1\;1]\T$ is the center of the circular obstacle and $R_o=0.5$ its radius, which has relative degree 2 for \eqref{eq:dbl_int} with respect to both $u$ and $\theta$ as required by Assumption \ref{assuption:relative-degree}. This constraint function is used to construct an ISSf-HOCBF candidate with $\alpha_1(s)=s$, $\alpha_2(s)=\tfrac{1}{2}s$, and $\varepsilon_h=1$.

    For each simulation, the closed-loop trajectory is generated by the controller in \eqref{eq:ISSf-ISS-QP}, the results of which are provided in Fig. \ref{fig:avg_error} and Fig. \ref{fig:traj}. As shown in Fig. \ref{fig:traj}, the trajectories under each update law remain safe and converge to the origin, whereas Fig. \ref{fig:avg_error} illustrates the convergence of the parameter estimation error to zero for each estimation algorithm as predicted by Lemma \ref{lemma:estimator}. The curves in Fig. \ref{fig:avg_error} represent the mean and standard deviation of the parameter estimation error over time across all simulations for each estimation algorithm. The results in Fig. \ref{fig:avg_error} illustrate that, on average, the RLS with forgetting factor estimator \eqref{eq:RLS-beta} produces the fastest convergence of the parameters estimates while also exhibiting low variance across different trajectories. The standard RLS algorithm \eqref{eq:RLS} produces the slowest convergence, which is expected given that, in general, this algorithm cannot guarantee exponential convergence\footnote{This also implies that \eqref{eq:RLS} does not satisfy all the conditions of Lemma \ref{lemma:estimator}. Despite this, note that boundedness of the estimates is sufficient to establish ISS and ISSf.} of the parameter estimates, whereas the others can \cite[Ch. 8.7]{Slotine}.

    \begin{figure}
        \centering
        \includegraphics{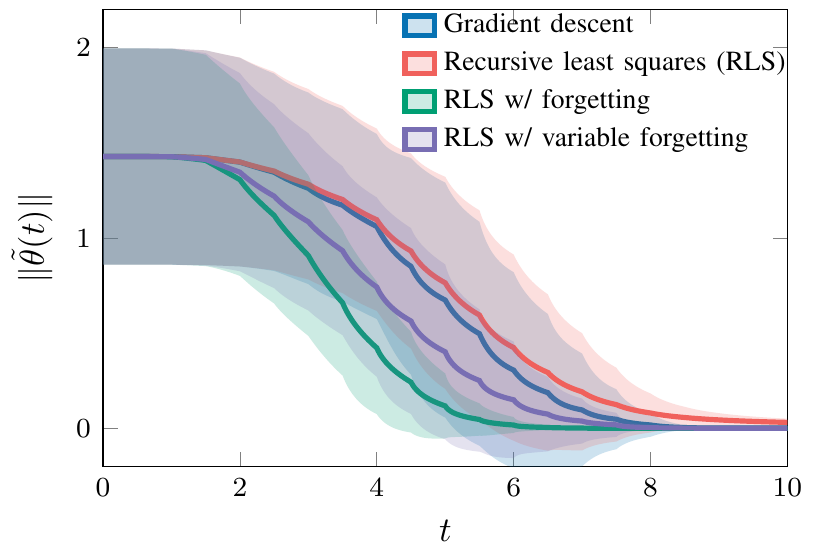}
        \caption{Mean and standard deviation of the norm of the parameter estimation over time generated by each parameter estimator. The solid lines indicate the average value of $\|\tilde{\theta}(t)\|$ across each simulation, and the ribbon surrounding each line corresponds to one standard deviation from the mean.}
        \label{fig:avg_error}
    \end{figure}

    \begin{figure}
        \centering
        \includegraphics{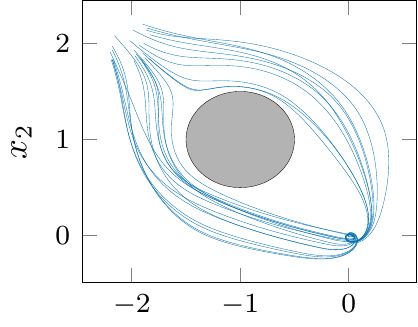}
        \hfill
        \includegraphics{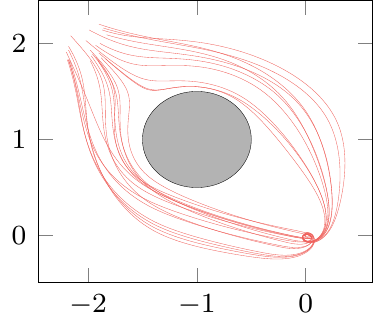}

        \includegraphics{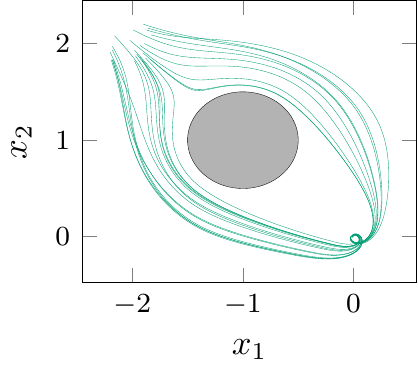}
        \hfill
        \includegraphics{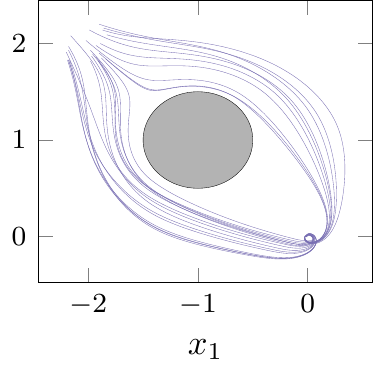}
        \caption{State trajectories generated by each estimation algorithm projected onto the $x_1$-$x_2$ plane. In each plot the gray disk denotes the obstacle. The colors in each plot share the same interpretation as those in Fig. \ref{fig:avg_error}.}
        \label{fig:traj}
    \end{figure}

    In the preceding examples, safety was enforced by choosing an appropriate value of $\varepsilon_h$ for the given level of uncertainty. In theory, $\mathcal{C}_\delta\rightarrow\mathcal{C}$ as $\varepsilon_h\rightarrow 0$; however, taking $\varepsilon_h$ very small may require a significant amount of control effort that could exceed physical actuator limits. An alternative approach to reducing safety violations in this ISSf setting is through fast adaptation - if the parameter estimates quickly converge to their true values then the estimated dynamics used in \eqref{eq:ISSf-ISS-QP} to generate control actions will be very close to the true dynamics. In Fig. \ref{fig:learning}, we generated additional trajectories of the closed-loop system under the gradient descent update law \eqref{eq:GD} and the RLS update law with a forgetting factor \eqref{eq:RLS-beta} using the same setup as in the previous example, but with different levels of initial parameter uncertainty. As demonstrated in Fig. \ref{fig:learning}, the trajectories under the RLS update law avoid the obstacle for the given initial parameter estimation errors via fast adaptation, whereas the trajectories under the gradient descent algorithm violate the safety constraint for higher levels of uncertainty. Hence, rather than using a more robust controller (by decreasing $\varepsilon_h$), which may be overly conservative if bounds on $\theta$ are unknown, one can endow the ISSf controller with stronger safety guarantees through the use of a more efficient estimation algorithm.

    \begin{figure}
        \centering
        \includegraphics{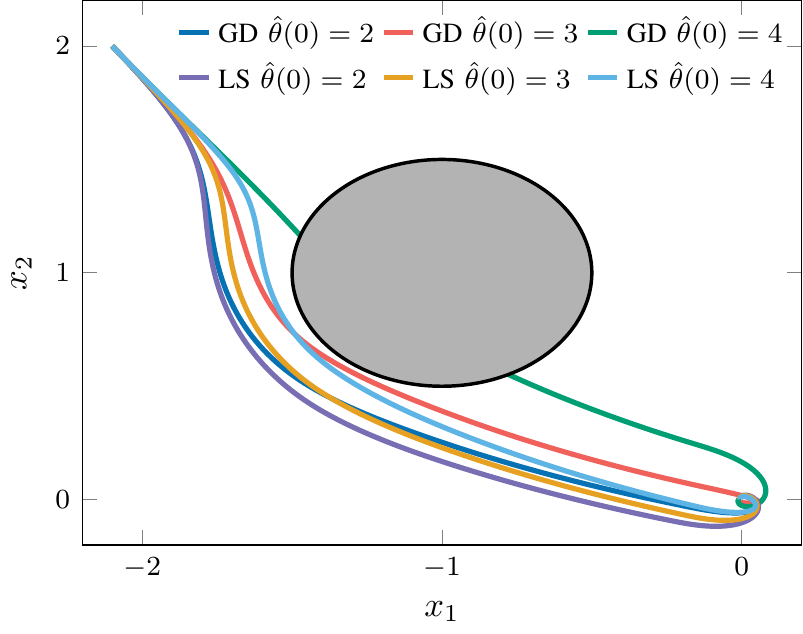}

        \includegraphics{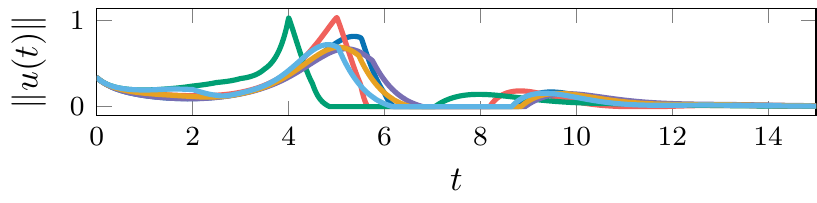}

        \includegraphics{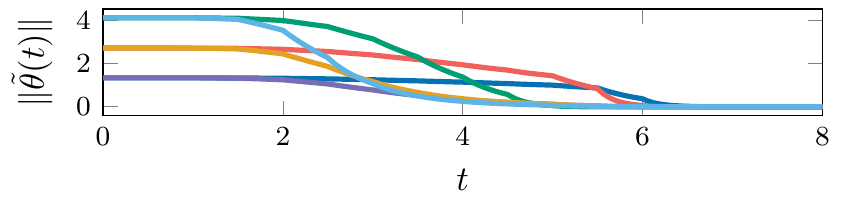}
        \caption{Comparison between trajectories generated by the gradient descent (GD) learning algorithm \eqref{eq:GD} and the recursive least squares algorithm (LS) with a forgetting factor \eqref{eq:RLS-beta} for different initial parameter estimates. The top plot displays the system trajectories, the middle illustrates the control trajectories, and the bottom illustrates the parameter estimation error.}\label{fig:learning}
    \end{figure}

    \section{Conclusion}\label{sec:conclusion}
    We presented a modular approach to safe adaptive control using CLFs and CBFs. In particular, we unified the concepts of ISS and ISSf to allow for freedom in the estimation algorithm used to learn the uncertain parameters while maintaining ISS and ISSf guarantees. Our hope is that this work facilitates the application of more advanced techniques from the machine learning literature \cite{RechtJMLR21} to parameter estimation and learning in a safety-critical setting.
    
    We showed in Sec. \ref{sec:stability} that convergence of the uncertain parameters is sufficient for asymptotic stability despite the use of an ISS-CLF. Based on the empirical results in this paper, it seems possible that a similar phenomenon may arise in the context of safety. Estimators that exhibited faster convergence maintained safety in the presence of large uncertainty when those with slower convergence did not - formally investigating the relation between convergence and safety thus remains an open research direction.

    \section*{Acknowledgements}
    We thank the anonymous reviewers for their careful reading of this paper and their helpful comments that have improved the quality of this work.

    \bibliographystyle{ieeetr}
    \bibliography{
        biblio/adaptive,
        biblio/books,
        biblio/barrier,
        biblio/hybrid,
        biblio/mpc,
        biblio/learning,
        biblio/nonlinear}
\end{document}